\documentclass[11p]{article}
\usepackage{latexsym,amsfonts,amssymb,amsmath,amscd,epic,graphics}
\usepackage{psfrag,epsfig}
\usepackage{graphicx}
\usepackage{color}  
\usepackage{subfigure}
\usepackage{psfrag,xspace}
\usepackage{multirow}

\usepackage{algorithmic} 
\usepackage[linesnumbered,vlined,boxruled,boxed,algo2e]{algorithm2e}
\usepackage{ifthen}
\usepackage{lscape}
\usepackage{marvosym}

\usepackage{vector}
\usepackage{url}
\usepackage{enumerate}

\usepackage{ccfonts}

\usepackage{notations}

\newtheorem{theorem}{Theorem}
\newtheorem{lemma}[theorem]{Lemma}

\newtheorem{remark}[theorem]{Remark}

\title{Improved bounds for Continued Fractions variants for real root
  isolation}

\author{Elias P.~Tsigaridas \\
  Computer Science Department, Aarhus University, Denmark\\
  \texttt{elias (AT) cs.au.dk}
}

\date{}

\begin{document}

\maketitle

\begin{abstract}
  We consider the problem of isolating the real roots of a square-free
  polynomial with integer coefficients using (variants of) the
  continued fraction algorithm (CF).
  We introduce a novel way to compute a lower bound on the positive
  real roots of univariate polynomials.  This allows us to derive a
  worst case bound of $\sOB( d^6 + d^4\tau^2 + d^3\tau^2)$ for
  isolating the real roots of a polynomial with integer coefficients
  using the classic variant \cite{Akritas:implementation} of CF, where
  $d$ is the degree of the polynomial and $\tau$ the maximum bitsize
  of its coefficients.  This improves the previous bound of Sharma
  \cite{sharma-tcs-2008} by a factor of $d^3$ and matches
  the bound derived by Mehlhorn and Ray \cite{mr-jsc-2009} for another
  variant of CF; it also matches the worst case bound of the
  subdivision-based solvers.

  % \vspace{10pt}
  
  % {\it The improved bound bound that was claimed in an earlier version is removed, 
  %   since there was an error in the proof.}

  %%
\end{abstract}

% \vspace{0.9mm}
% \noindent
% {\bf Categories and Subject Descriptors:}
% F.2 [Theory of Computation]: Analysis of Algorithms and Problem Complexity;
% I.1 [Computing Methodology]: Symbolic and algebraic manipulation: Algorithms

% \vspace{5pt}
% {\bf Keywords} real root isolation, continued fraction,
% real root problem, separation bound, polynomial, Boolean complexity 

\section{Introduction}

The problem of isolating the real roots of a square-free polynomial
with integer coefficients is one of the most well-studied problems in
symbolic computation and computational mathematics. The goal is to
compute intervals with rational endpoints that contain one and only
one real root of the polynomial, and to have one interval for every
real root.

If we restrict ourselves to algorithms that perform computations with
rational numbers of arbitrary size, then we can distinguish two main
categories.  The first one consists of algorithms that are
subdivision-based; their process mimics binary search.  They bisect an
initial interval that contains all the real roots until they obtain
intervals with one or zero real roots.  The different variants differ
in the way that they count the number of real roots inside an
interval, for example using Sturm's theorem or Descartes' rule of
signs, see also Th.~\ref{th:Descartes}. Classical representatives are
the algorithms \func{sturm}, \func{descartes} and \func{bernstein}.
We refer the reader to
\cite{ESY:descartes,emt-lncs-2006,Dav:TR:85,KraMeh:descartes:05,Krandick:Isolation,Johnson-phd-91,Yap:SturmBound:05,RouZim:solve:03}
and references therein for further details. The worst case complexity
of all variants in this category is $\sOB(d^6 + d^4 \tau^2)$, where
$d$ is the degree of the polynomial and $\tau$ the maximum bitsize of
its coefficients. Especially, for the \func{sturm} solver, recently,
it was proved that its expected case complexity, if we consider
certain random polynomials as input, is $\sOB(r \,d^2 \tau)$, where
$r$ is the number of real roots \cite{egt-issac-2010}.  Let us also
mention the bitstream version of \func{descartes} algorithm,
cf. \cite{ms-jsc-2010} and references therein.

The second category contains algorithms that isolate the real roots of
a polynomial by computing their continued fraction expansion (CF).
Since successive approximants of a real number
define an interval that contains this number, CF computes the
partial quotients of the roots of the polynomial until the
corresponding approximants correspond to  intervals that isolate the real
roots.  Counting of the real roots is based on Descartes' rule of
signs (Th.~\ref{th:Descartes}) and termination is guaranteed by
Vincent's theorem (Th.~\ref{th:Vincent}).
There are several variants which they differ in the way that they
compute the partial quotients.
%There are several variants which they differ in the way that they
%compute the partial quotients in the continued fraction expansion of the real roots.

The first formulation of the algorithm is due to Vincent
\cite{Vincent}, who computed the partial quotients by successive
transformations of the form $x \mapsto x+1$. An upper bound on the
number of partial quotients needed was derived
by Uspensky \cite{u-te-48}. Unfortunately this approach leads to an
exponential complexity bound.  Akritas \cite{Akritas:implementation},
see also \cite{AkrStrVig-computing-2006,Akritas:NoUspensky}, treated
the exponential behavior of CF by treating the partial quotients as
lower bounds of the positive real roots, and computed the bounds using
Cauchy's bound. With this approach, $c$ repeated operations of the
form $x \mapsto x +1$ could be replaced by $x \mapsto x + c$. However,
his analysis assumes an ideal positive lower bound, that is that we
can compute directly the floor of the smallest positive real root.
% and it is not clear how to take into account the increased coefficient
% size of the transformed polynomial.  
In \cite{te-esa-2006}, it was proven, under the assumption that
Gauss-Kuzmin distribution holds for the real algebraic numbers, that
the expected complexity of CF is $\sOB( d^4 \tau^2)$.  By spreading
the roots, the expected complexity becomes $\sOB( d^4 + d^3 \tau)$
\cite{te-tcs-2008}.  The first worst-case complexity result of CF,
$\sOB( d^8 \tau^3)$, is due to Sharma \cite{sharma-tcs-2008}, without
any assumption.  He also proposed a variant of CF, that combines
continued fractions with subdivision, with complexity $\sOB(d^5
\tau^2)$. All the variants of CF in \cite{sharma-tcs-2008} compute
lower bounds on the positive roots using Hong's bound
\cite{Hong:jsc:98}, which is assumed to have quadratic arithmetic
complexity.  Mehlhorn and Ray \cite{mr-jsc-2009} proposed a novel way
of computing Hong's bound based on incremental convex hull
computations with linear arithmetic complexity.  A direct consequence
is that they reduced the complexity of the variant of CF combined with
subdivision \cite{sharma-tcs-2008} to $\sO( d^4 \tau^2)$, thus
matching the worst case complexity of the subdivision-based
algorithms. Using \cite{mr-jsc-2009} and fast Taylor shifts
\cite{GatGer:fast_shift:97}, the bound \cite{sharma-tcs-2008} on
classical variant of CF becomes $\sOB( d^7 \tau^3)$.

As far as the numerical algorithms are concerned, the best known bound
for the problem is due to Pan \cite{Pan02jsc,Pan97rev} and Sch\"onhage
\cite{Sch82}, see also \cite{s-icm-1986}, $\sOB(d^3\tau)$.  Moreover,
it seems that Pan's approach could be improved to $\sOB(d^2\tau)$.
This class of algorithms approximate the roots, real and complex, of
the input polynomials up to a precision. They could be turned to root
isolation algorithms by requiring them to approximate up to the
separation bound, that is the minimum distance between the roots.  The
crux of the algorithms is that they recursively split the polynomial
until we obtain linear factors that approximate sufficiently all the
roots, real and complex.  We also refer to a recent approach that
concentrates only on the real roots \cite{pan-snc-2007}.
For an implementation of Sch\"onhage's algorithm we refer the reader 
to the routine CPRTS, p.12 in Addenda, based on the multitape Turing machine%
\footnote{\url{http://www.iai.uni-bonn.de/~schoe/tp/TPpage.html}}.
We are not aware of any implementation of Pan's algorithm.
In the special case where all the roots of the polynomial are real,
also called the {\em real root problem}, dedicated numerical
algorithms were proposed by Reif \cite{r-focs-1993} and Ben-Or and
Tiwari \cite{bt-joc-1990} for approximating the roots.
%, and also effective parallel versions \cite{bfkt-hyperbolic-solve}.
Nevertheless, their Boolean complexity is also $\sOB(d^3 \tau)$.
Quite recently, Sagraloff \cite{s-arxiv-isol-10} announced a variant
of the bitstream version of \func{descartes} algorithm with complexity
$\sOB(d^3\tau^2)$.

% There is a huge amount of work on the problem of isolating or
% approximating the (real) roots of a polynomial. The presented
% references represent only a small part of it. For this we encourage
% the reader to refer to the references.

{\em Our contribution.}
We present a novel way to compute a lower bound on the positive real
roots of a univariate polynomial (Lem.~\ref{lem:plb-computation}).
The proposed approach computes the floor of the root (possible
complex) with the smallest positive real part that contributes to the
number of the sign variations in the coefficients list of the
polynomial.  Our bound is at least as good as Hong's bound
\cite{Hong:jsc:98}.  Using this lower bound we improve the
worst case bit complexity bound of the classical variant of CF,
obtained by Sharma \cite{sharma-tcs-2008}, by a factor of $d^3$. We
obtain a bound of $\sOB( d^6 + d^4 \tau^2 )$ or $\sOB(N^6)$, where
$N=\max\{d, \tau\}$, (Th.~\ref{th:CF-bound}), which matches the worst
case bound of the subdivision-based solvers and also matches the bound
due to Mehlhorn and Ray \cite{mr-jsc-2009} achieved for another
variant of CF; it also matches the worst case bound of the
subdivision-based solvers
\cite{ESY:descartes,emt-lncs-2006,Dav:TR:85,KraMeh:descartes:05,Krandick:Isolation,Johnson-phd-91,Yap:SturmBound:05,RouZim:solve:03}.
%%

% \paragraph*{Paper Structure}
% The rest of the paper is structured as follows.
% First we specify our notation.
% Sec.~\ref{sec:intro-cf} presents a short introduction to the theory of continued fractions.
% In Sec.~\ref{sec:classic-cf} we present the algorithm to compute lower bounds and we derive the worst case complexity bound of CF.
% In Sec.~\ref{sec:new-variant} we present iCF and its complexity analysis.

{\em Notation.}
In what follows \OB, resp. \OO, means bit, resp. arithmetic,
complexity and the \sOB, resp. \sO, notation means that we are
ignoring logarithmic factors.
For a polynomial $A \in$ $\ZZ[x]$,
$\dg{A} = d$ denotes its degree and $\bitsize{A} =\tau$ the maximum bitsize of
its coefficients, including a bit for the sign. 
For $a \in \QQ$, $\bitsize{ a} \ge 1$
is the maximum bitsize of the numerator and the denominator.
Let $\Multiply{\tau}$ denote the bit complexity of multiplying two integers of size
$\tau$; using \textsc{FFT}, $\Multiply{\tau} = \sOB( \tau )$.
To simplify notation, we will assume throughout the paper that 
$\lg(\dg{A}) = \lg{d} = \OO(\tau) = \OO( \bitsize{A})$.
By $\var(A)$ we denote the number of sign variations in the list of
coefficients of $A$.  We use $\Delta_{\gamma}$ to denote the minimum
distance between a root $\gamma$ of a polynomial $A$ and any other
root, we call this quantity {\em local separation bound};
$\Delta=\min_{\gamma}{\Delta_{\gamma}}$ is the {\em separation bound},
that is the minimum distance between all the roots of $A$.

\section{A short introduction to continued fractions}
\label{sec:intro-cf}

Our presentation follows closely \cite{te-tcs-2008}.
For additional details we refer the reader to, e.g., \cite{Yap2000,BomPoo:contfrac:95,Poorten:intro}.
In general, a {\em simple (regular) continued fraction} 
is a (possibly infinite)  expression of the form 
\begin{displaymath}
  q_0 +
  \cfrac{1}{
    q_1 + \cfrac{1}{
      q_2 + \dots  %+ \cfrac{1}{c_n}
    }
  } =
  [ q_0, q_1, q_2, \dots ],
\end{displaymath}
where the numbers $q_i$ are called {\em partial quotients}, 
$q_i \in \ZZ$ and $q_i \geq 1$ for $i > 0$.  
Notice that $q_0$ may have any sign, however, in our real root isolation
algorithm $q_0 \geq 0$, without loss of generality.
By considering the recurrent relations
\begin{equation}
  \begin{array}{cccc}
    P_{-1} = 1, & P_{0} = q_0, & P_{n+1} = q_{n+1}\, P_{n} + P_{n-1},\\
    Q_{-1} = 0, & Q_{0} = 1,   & Q_{n+1} = q_{n+1}\, Q_{n} + Q_{n-1},
  \end{array}
  \label{eq:recur-PQ}
\end{equation}
it can be shown by induction that $R_n = \frac{P_n}{Q_n} = [q_0, q_1, \dots, q_n]$,
for $n=0,1,2,\dots$.
% and moreover that
% \begin{displaymath} 
%   \begin{array}{lcl}
%     P_{n}\, Q_{n+1} - P_{n+1}\, Q_{n} & = & (-1)^{n+1},\\
%     P_{n}\, Q_{n+2} - P_{n+2}\, Q_{n} & = & (-1)^{n+1} c_{n+2}.
%   \end{array}
% \end{displaymath}

If $\gamma = [q_0, q_1, \dots]$ then 
$
\gamma = q_0 + \frac{1}{Q_0 Q_1} - \frac{1}{Q_1 Q_2} + \dots
= q_0 + \sum_{n=1}^{\infty}{ \frac{(-1)^{n-1}}{Q_{n-1}Q_n}}
$ 
and since this is a series of decreasing alternating terms it converges to some real
number $\gamma$.
A finite section $R_n = \frac{P_n}{Q_n} = [ q_0, q_1, \dots, q_n]$
is called the $n-$th {\em convergent} (or {\em approximant}) of $\gamma$
and the tails $\gamma_{n+1} = [q_{n+1}, q_{n+2}, \dots]$  are known as its
{\em complete quotients}. 
That is $\gamma = [ q_0, q_1, \dots, q_n,$ $ \gamma_{n+1}]$
for $n=0,1,2,\dots$.
There is an one to one correspondence between the real numbers and the continued
fractions, where evidently the finite continued fractions correspond
to rational numbers.

It is known that $Q_n \geq F_{n+1}$ 
and that $F_{n+1} < \phi^n < F_{n+2}$, 
where $F_n$ is the $n-$th Fibonacci number 
and $\phi = \frac{1+ \sqrt{5}}{2}$ is the golden ratio.
Continued fractions are the best rational approximation (for a given denominator size).
This is as follows:
$ \frac{1}{Q_n(Q_{n+1} +Q_n)} 
  \leq \left| \gamma - \frac{P_n}{Q_n} \right| 
  \leq \frac{1}{Q_n Q_{n+1}} < \phi^{-2n+1}.
$

In order to indicate or to emphasize that a partial quotient or an
approximant belong to a specific real number $\gamma$, we use the
notation $q_i^{\gamma}$ and $R_n^{\gamma}
=P_n^{\gamma}/Q_n^{\gamma}$, respectively.  We also use $q_i^{(k)}$
and $R_n^{(k)} =P_n^{(k)}/Q_n^{(k)}$, where $k$ is a non-negative
integer, to indicate that we refer to the (real part of the) root
$\gamma_k$ of a polynomial $A$.  The ordering of the roots is
considered with respect to the magnitude of their (positive) real part.

\section{Worst case complexity of CF}
\label{sec:classic-cf}

\begin{theorem}[Descartes' rule of sign]
  \label{th:Descartes}
  The number $R$ of real roots of $A(x)$ in $(0, \infty)$ is bounded by $\var(A)$ and
  we have $R \equiv \var(A) \mod 2$.
\end{theorem}

 \begin{remark}
   \label{rem:Descartes}
   In general Descartes' rule of sign obtains an overestimation of the number of the
   positive real roots.
   However, if we know that $A$ is {\em hyperbolic}, i.e. has only real
   roots, or when the number of sign variations is 0 or 1 then it counts exactly.
 \end{remark}

The CF algorithm depends on the following theorem, which 
dates back to Vincent's theorem in 1836 \cite{Vincent}.
The inverse of Th.~\ref{th:Vincent} can be found in
\cite{Akri89,ColLoo82,Mign91}.
The version of the theorem that we present is due to Alesina and
Galuzzi \cite{galuzzi98:vincent}, 
see also \cite{u-te-48,Akritas:implementation,Akri89,Akritas:NoUspensky},
and its proof is closely connected to the one and two circle theorems 
(refer to \cite{KraMeh:descartes:05,galuzzi98:vincent} and references therein).
\begin{theorem}
  \label{th:Vincent}
  {\rm \cite{galuzzi98:vincent}}
  Let  $A \in \ZZ[x]$ be square-free and let $\Delta > 0$ be the
  separation bound, i.e. the smallest distance between two (complex) roots of $A$.
  Let $n$ be the smallest index such that  
  $F_{n-1}\,F_{n}\, \Delta > \frac{2}{\sqrt{3}}$,
  where $F_n$ is the $n$-th Fibonacci number.
  Then the map $x \mapsto [c_0 , c_1, \dots, c_n, x]$,
  where $c_0, c_1, \dots, c_n$ is an arbitrary sequence of positive integers,
  transforms $A(x)$ to $A_n( x)$, whose list of coefficients has no more than one sign variation.
\end{theorem}

For a polynomial $A=\sum_{i=0}^{d}{a_ix^i}$, 
where $\gamma$ correspond to its (complex) roots, 
the Mahler measure, $\Mahler{A}$, of $A$ is 
$\Mahler{A} = a_d \prod_{\abs{\gamma} \geq 1}{\abs{\gamma}}$,
e.g. \cite{Mign91,Yap2000}.
If we further assume that $A\in \ZZ[x]$ and $\bitsize{A}=\tau$
then 
$\Mahler{A} \leq \norm{A}_2 \leq \sqrt{d+1} \norm{A}_{\infty} 
= 2^{\tau} \sqrt{d+1} $,
and so $\prod_{\abs{\gamma} \geq 1}{\abs{\gamma}} \leq 2^{\tau} \sqrt{d+1}$.

We will also use the following  aggregate bound.
For a proof we refer to e.g.~\cite{te-tcs-2008,Dav:TR:85,Yap:SturmBound:05,Mign91,Johnson-phd-91}.
\begin{theorem}
  \label{th:dmm-1}
  Let $A \in \ZZ[x]$ such that $\dg{A}=d$ and $\bitsize{A}=\tau$.
  Let $\gamma$ denotes its distinct roots, then 
  \begin{displaymath}
  \begin{aligned}
    &\prod_{\gamma}{ \Delta_{\gamma}} \geq 2^{-d^2} \, \Mahler{A}^{-2d}
    \Leftrightarrow -\lg{ \prod_{\gamma}{ \Delta_{\gamma}}} = 
    -\sum_{\gamma}{\lg{\Delta_{\gamma}}} \leq 3 d^2 + 3 d \lg{d} + 3 d \tau.
  \end{aligned}
\end{displaymath}

\end{theorem}

\subsection{The tree}

The CF algorithm relies on Vincent's theorem
(Th.~\ref{th:Vincent}) and Descartes' rule of sign
(Th.~\ref{th:Descartes}) to isolate the positive real roots of 
a square-free polynomial $A$. 
The negative roots are isolated after we perform the transformation
$x \mapsto -x$; hence it suffices to consider only the case of
positive real roots throughout the analysis.

The pseudo-code of the classic variant of CF is presented in
Alg.~\ref{alg:CF}.

Given a polynomial $A$, we compute the floor of the smallest positive
real root (\func{plb} = Positive Lower Bound).  The {\em ideal}
\func{plb} is a function that can determine whether a polynomial has
positive real roots, and if there are such roots then returns the
floor of the smallest positive root of the polynomial.

Then we perform the transformation $x \mapsto x + \rat{b}$, obtaining a
polynomial $A_{\rat{b}}$.
It holds that $\var(A) \geq \var(A_{\rat{b}})$.
The latter polynomial is transformed to $A_1$ 
by the transformation $x \mapsto 1 + x$
and if $\var(A_1) = 0$ or $\var(A_1) = 1$,
then $A_{\rat{b}}$ has 0, resp. 1, real roots greater than 1,
or equivalently
$A$ has 0, resp. 1, real roots greater than $\rat{b}+1$ 
(Th.~\ref{th:Descartes}).
If $\var(A_1) < \var(A_{\rat{b}})$ then (possibly)
there are real roots of $A_{\rat{b}}$ in $(0, 1)$,
or equivalently, 
there are real roots of $A$ in $(\rat{b}, \rat{b}+1)$,
due to Budan's theorem.
We apply the transformation $x \mapsto 1/(1+x)$ to $A_{\rat{b}}$,
and we get the polynomial $A_2$.
If $\var(A_2) = 0$ or $\var(A_2) = 1$,
$A_{\rat{b}}$ has 0, resp. 1,
real root less than 1 (Th.~\ref{th:Descartes}),
or equivalently
$A$ has 0, resp. 1, real root less than $\rat{b}+1$,
or to be more specific in $(\rat{b}, \rat{b}+1)$ (Th.~\ref{th:Descartes}).
If the transformed polynomial, $A_1$ and $A_2$,
have more than one sign variations, then we apply \func{plb} to them
and we repeat the process.

Following 
\cite{Akritas:implementation,te-tcs-2008,sharma-tcs-2008} we consider
the process 
of the algorithm as an infinite binary tree. The nodes of the tree
hold polynomials and (isolating) intervals.
The root of the tree corresponds to the original polynomial $A$ and
the shifted polynomial $A_{\rat{b}}$. 
The branch from a node to a right child corresponds to the map 
$x \mapsto x+1$, which yields polynomial $A_1$,
while to the left child to the map $x \mapsto 1/(1+x)$,
which yields polynomial $A_2$.
The sequence of transformations that we perform is equivalent to the
sequence of transformations in Th.~\ref{th:Vincent}, 
and so the leaves of the tree hold (transformed) polynomials that have
no more than one sign variation, if Th.~\ref{th:Vincent} holds.

A polynomial that corresponds to a leaf of the tree and
has one sign variation it is produced after a transformation as in
Th.~\ref{th:Vincent}, using positive integers $q_0, q_1, \dots, q_n$.
The compact form of this is 
$M : x \mapsto \frac{P_n x + P_{n-1}}{Q_n x + Q_{n-1}}$,
where $\frac{P_{n-1}}{Q_{n-1}}$ and $\frac{P_n}{Q_n}$ are consecutive
convergents of the continued fraction $[q_0, q_1, \dots, q_n]$.
The polynomial has one real root in
$(0,\infty)$, thus the (unordered) endpoints of the isolating interval are 
$M(0) = \frac{P_{n-1}}{Q_{n-1}}$ and $M(\infty) = \frac{P_n}{Q_n}$.

There are different variants of the algorithm that differ in the way
they compute \func{plb}.
A \func{plb} realization that actually computes exactly the floor of the
smallest positive real root is called {\em ideal}, but unfortunately
has a prohibitive complexity. 

A crucial observation is that Descartes' rule of sign
(Th.~\ref{th:Descartes}), that counts the number of sign variations
depends not only on positive real roots, but also on some complex
ones; which have positive real part. 
Roughly speaking CF is trying to isolate the positive real parts of the roots of
$A$ that contribute to the sign variations. 
Thus, the ideal \func{plb} suffices to compute the floor of the smallest
positive real part of the roots of $A$ that contribute to the number
of sign variations.  For this we will use
Lem.~\ref{lem:plb-computation}.  Notice that all the positive real
roots contribute to the number of sign variation of $A$, but this is
not always the case for the complex roots with positive real part.

\subsection{Computing a partial quotient}

\begin{lemma}
  \label{lem:plb-computation}
  Let $A \in \ZZ[x]$, such that $\dg{A} = d$ and $\bitsize{A}=\tau$.
  We can compute the first partial quotient, or in the other words the
  floor%
  \footnote{{\scriptsize We choose to use $c$ instead of $q_0$ because
      in the complexity analysis that follow $A$ could be a result of
      a shift operation, thus the computed integer may not be the
      0-th partial quotient of the root that we are trying to
      approximate.}}, 
  $c$,
  of the real part of the root with the smallest
  real part, that contributes to the sign variations of $A$ in 
  $\sOB(d \tau \lg{c} + d^2 \lg^2{c})$.
\end{lemma}
\begin{proof}
  We compute the corresponding integer using the technique of the
  exponential search, see for example \cite{KweMeh-IPL-2003}.  Without
  loss of generality, we may assume that the real root is not in
  $(0,1)$, since in this case we should return 0.
  
  We perform the transformation $x \mapsto x + 2^0$ to the polynomial,
  and then the transformation $x \mapsto x + 1$. 
  If the number of sign variations of the resulting polynomial
  compared to the original one decreases,
  then  $2^0 = 1$ is the partial quotient. 
  If not, then we perform the transformation $x \mapsto x + 2^1$. If the number of sign variations
  does not decrease, then we perform $x \mapsto x + 2^2$. 
  Again if the number of sign variations does not decrease, then we perform 
  $x \mapsto x + 2^3$ and so on.
  Eventually, for some positive integer $k$, 
  there would be a loss in the sign variations between transformations 
  $x \mapsto x + 2^{k-1}$ and $x \mapsto x +2^{k}$.
  In this case the partial quotient $c$, which we want to compute,
  satisfies  $2^{k-1} < c < 2^{k} < 2\, c$.
  The exact value of $c$ is computed by performing binary search in the interval
  $[2^k,  2^{k+1}]$. 
  We deduce that the number of transformations that we need to perform is 
  $2k + \OO(1) = 2 \lg\floor{c} + \OO( 1)$.
  
  In the worst case, each transformation corresponds to an
  asymptotically fast Taylor shift with a number of bitsize $\OO( \lg{c})$,
  which costs% 
  \footnote{{\scriptsize Following Th.~2.4(E) in \cite{GatGer:fast_shift:97} 
      the cost of performing the operation $f(x+a)$, where $\dg{f}=n$, 
      $\bitsize{f}=\tau$ and $\bitsize{a}=\sigma$ is
      $\OB( \Multiply{n\tau + n^2\sigma)\lg{n}})$, 
      and if we assume fast multiplication algorithms between integers, 
      then it becomes
      $\sOB(n\tau + n^2\sigma)$.}}
  $\OB( \Multiply{d \tau + d^2 \lg{c}}\lg{d})$
  \cite[Th.~2.4]{GatGer:fast_shift:97}.
  By considering fast multiplication algorithms 
  the costs becomes  $\sOB( d \tau + d^2 \lg{c})$
  and multiplying by the number of transformations needed,
  $\lg{c}$, we conclude the proof.
    
  It is worth noticing that we do not consider 
  the cases $c = 2^{k}$ or $c = 2^{k+1}$, 
  since then we have computed, exactly, a rational root.
\end{proof}

\subsection{Shifts operations and total complexity}

Up to some constant factors, we can replace $\Delta$ in
Th.~\ref{th:Vincent} by $\Delta_{\gamma}$, refer to
\cite{sharma-tcs-2008} for a proof. This allows us to estimate the
number, $m_{\gamma}$, of partial quotients needed, in the worst case,
to isolate the positive real part of a root $\gamma$. It holds 
\begin{displaymath}
  m_{\gamma} 
  \leq \frac{1}{2}(1 + \log_{\phi}{2} -  \lg{ \Delta_{\gamma}})
  \leq 2 - \frac{1}{2} \lg{ \Delta_{\gamma}}.
\end{displaymath}
The transformed polynomial has either one
or zero sign variation and if $\gamma \in \RR$, then the corresponding
interval isolates $\gamma$ from the other roots of $A$.
The associated continued fraction of (the real part of) $\gamma$ is
$[q_0^{\gamma}, q_1^{\gamma}, \dots, q_{m_{\gamma}}^{\gamma}]$.
It holds that $\sum_{\gamma}{m_{\gamma}} = \OO( d^2 + d \tau)$ \cite{te-tcs-2008,sharma-tcs-2008}.
The following lemma bounds the bitsize of the partial quotients,
$q_k^{\gamma}$, of a root $\gamma$.

\begin{lemma}
  \label{lem:q-bounds}
  Let $A \in \ZZ[x]$, such that $\dg{A}=d$ and $\bitsize{A}=\tau$.
  For the real part of any root $\gamma$ it holds 
  \begin{displaymath}
    \sum_{j=0}^{m_{\gamma}}{ \lg( q_j^{\gamma})} 
    = \lg(q_0^{\gamma}) + \sum_{j=1}^{m_{\gamma}}{ \lg( q_j^{\gamma})} 
    \leq \lg(q_0^{\gamma}) + 1 - \lg{\Delta_{\gamma}},
  \end{displaymath}
  where we assume that $q_0^{\gamma} > 0$,
  and the term $1 - \lg{\Delta_{\gamma}}$ appears only when 
  $\Delta_{\gamma} < 1$, i.e. when $m_{\gamma} \geq 1$.
  Moreover 
  $\sum_{\gamma}{ \lg( q_0^{\gamma})} \leq \lg{\norm{A}_2} \leq \tau + \lg{d}$
  and if $\gamma$ ranges over a subset of distinct roots of $A$, then 
  \begin{displaymath}
    \sum_{\gamma} \sum_{k=0}^{m_{\gamma}} \lg{q_k^{\gamma}}
     \leq 1 + \tau + \lg{d} - \lg{ \prod_{\gamma}{ \Delta_{\gamma}}} 
     = \OO( d^2 + d \tau).
   \end{displaymath}  
\end{lemma}
\begin{proof}
  The Mahler measure, $\Mahler{A}$, of
  $A$ is 
  $\Mahler{A} = a_d \prod_{\abs{\gamma} \geq 1}{\abs{\gamma}}$.
  It also holds 
  $\Mahler{A} \leq \norm{A}_2 \leq \sqrt{d+1} \norm{A}_{\infty} 
  = 2^{\tau} \sqrt{d+1} $,
  and so $\prod_{\abs{\gamma} \geq 1}{\abs{\gamma}} \leq 2^{\tau} \sqrt{d+1}$.
  Since $q_0^{\gamma}$ is the integer part of $\gamma$ it holds
  $\prod_{\gamma}{q_0^{\gamma}} \leq \prod_{\abs{\gamma} \geq 1}{\abs{\gamma}} 
  \leq \norm{ A}_{2}$
  and thus 
  \begin{equation}
    \sum_{\gamma}{ \lg( q_0^{\gamma})} \leq \lg{\sqrt{d+1}} 
    + \lg{ \norm{ A}_{\infty}}  \leq \tau + \lg{d}.
    \label{eq:q0-bound}
  \end{equation}
  Following  \cite{sharma-tcs-2008} we know that 
  \begin{equation}
    \frac{1}{Q_{m_{\gamma}}^{\gamma} Q_{m_{\gamma}-1}^{\gamma}} \geq \frac{\Delta_{\gamma}}{2}
    \Leftrightarrow
    Q_{m_{\gamma}}^{\gamma} Q_{m_{\gamma}-1}^{\gamma} \leq 2/\Delta_{\gamma}.
    \label{eq:QQ-D}
  \end{equation}
  From Eq.~(\ref{eq:recur-PQ}) we get 
  $Q_k = q_k Q_{k-1} + Q_{k-2} \Rightarrow Q_k \geq q_k Q_{k-1}$, for
  $k\geq 1$.
  I we apply the previous relation recursively we get
  $\prod_{k=1}^{m_{\gamma}}{ q_k^{\gamma}} \leq Q_{m_{\gamma}}^{\gamma}  \leq 2/\Delta_{\gamma}$
  and
  $\prod_{k=1}^{m-1}{ q_k^{\gamma} } \leq  Q_{m_{\gamma}-1}^{\gamma} \leq 2/\Delta_{\gamma}$,
  and so 
  \begin{displaymath}
    \sum_{k=1}^{m_{\gamma}}{ \lg{ q_k^{\gamma}}} = 
    \lg{ \prod_{k=1}^{m}{ q_k^{\gamma}}} 
    \leq 1 - \lg{\Delta_{\gamma}}.
  \end{displaymath}

  Finally, we sum over all roots $\gamma$ and 
  we use (\ref{eq:q0-bound}) and Th.~\ref{th:dmm-1},
  \begin{displaymath}
    \begin{aligned}
      \sum_{\gamma} \sum_{k=0}^{m_{\gamma}} \lg{q_k^{\gamma}} & =
      \sum_{\gamma} \lg{q_0^{\gamma}} 
      + \sum_{\gamma}  \sum_{k=1}^{m_{\gamma}} \lg{q_k^{\gamma}} 
       \leq \sum_{\gamma} \lg{q_0^{\gamma}} + \sum_{\gamma} (1 - \lg{\Delta_{\gamma}}) \\
      & \leq 1 + \tau + \lg{d} + d^2 + 3d\lg{d} + 3 d \tau,
    \end{aligned}
  \end{displaymath}
  which completes the proof.
\end{proof}

% \begin{remark}
%   \label{rem:Delta-less-than-1}
%   It is worth noticing that in the previous lemma it is implicitly
%   implied $\Delta_{\gamma} < 1$, which means that there is another,
%   possible complex, root in distance $<1$ to $\gamma$. This is so
%   since otherwise the root could be isolated without computing any
%   partial quotient, with the exception of $q_0^{\gamma}$.
% \end{remark}

At each step of CF we compute a partial quotient and we apply a Taylor
shift to the polynomial with this number.
In the worst case we increase the bitsize of the polynomial 
by an additive factor of $\OO( d \lg( q_k^{\gamma}))$,
at each step.
The overall complexity of CF is dominated by the computation of the
partial quotients.

The following table summarizes the costs of computing the partial
quotients of $\gamma$ that we need:
\begin{displaymath}
% \begin{array}{clr}
\begin{aligned}
    0^{th} \text{ step}  && \sOB( d \tau \lg(q_0^{\gamma}) +  d^2 \lg(q_0^{\gamma}) \lg(q_0^{\gamma}))\\\noalign{\vspace{3pt}}
    1^{st} \text{ step}  && \sOB( d \tau\lg(q_1^{\gamma}) + d^2 \lg( q_0^{\gamma} q_1^{\gamma}) \lg( q_1^{\gamma})) \\
    &&
    \Paren{= \sOB( d(\tau + d \lg( q_0^{\gamma}))\lg(q_1^{\gamma}) + d^2 \lg^2(q_1^{\gamma}))} \\\noalign{\vspace{3pt}}
    2^{nd} \text{ step}  && \sOB( d \tau\lg(q_2^{\gamma}) + d^2  \lg( q_0^{\gamma} q_1^{\gamma} q_2^{\gamma}) \lg( q_2^{\gamma}))\\\noalign{\vspace{3pt}}
    \vdots& \\\noalign{\vspace{3pt}}
    m_{\gamma}^{th} \text{ step}  && \sOB( d \tau \lg(q_m^{\gamma}) + d^2 \lg(\prod_{k=0}^{m}{q_k^{\gamma}}) \lg(q_m^{\gamma})) 
    &  \\ \noalign{\vspace{3pt}}
%     \hline
%     && \sOB\Paren{ 
%       d \tau \sum_{k=0}^{m}{\lg(q_k)} + 
%       d^2 \sum_{k=0}^{m}{\lg(q_{k}) \lg \prod_{j=0}^{m}{q_j}}} \\\noalign{\vspace{3pt}}
   % &&\mathcal{C}^{\gamma}_{1} + \mathcal{C}^{\gamma}_{2} = \mathcal{C}^{\gamma} 
\end{aligned}
%\end{array}
\end{displaymath}

We sum over all steps to derive the cost for isolating $\gamma$,
$\mathcal{C}^{\gamma}$,
and after applying some obvious simplifications and use
Lem.~\ref{lem:q-bounds} we get
\begin{displaymath}
  \begin{aligned}
   \mathcal{C}^{\gamma} & = \sOB\Paren{ 
      d \tau \sum_{k=0}^{m_{\gamma}}{\lg(q_k^{\gamma})} + 
      d^2 \sum_{k=0}^{m_{\gamma}}{\lg(q_{k}^{\gamma}) \lg \prod_{j=0}^{m_{\gamma}}{q_j^{\gamma}}}}  
     = \sOB\Paren{ d\tau \sum_{k=0}^{m_{\gamma}}{\lg(q_k^{\gamma})} +
      d^2\Paren{\sum_{k=0}^{m_{\gamma}}{\lg(q_k^{\gamma})}}^2 } \\
    & =
  \sOB\Paren{ d\tau (\lg(q_0^{\gamma}) - \lg{\Delta_{\gamma}}) + d^2 (\lg^2(q_0^{\gamma}) + \lg^2{\Delta_{\gamma}})}.
  \end{aligned}
\end{displaymath}

To derive the overall complexity, $\mathcal{C}$, 
we sum over all the roots that CF tries to isolate
and we use Lem.~\ref{lem:q-bounds} and Th.~\ref{th:dmm-1}. Then
{\footnotesize
\begin{equation}
  \label{eq:complexity-CF}
  \begin{array}{l}
    \mathcal{C} = 
    \sum_{\gamma}{\mathcal{C}^{\gamma}}  \vspace{3pt}\\
    = \sOB\Paren{ 
      d\tau \sum_{\gamma} \lg(q_0^{\gamma}) - d\tau\sum_{\gamma}\lg{\Delta_{\gamma}} 
      + d^2\sum_{\gamma} \lg^2(q_0^{\gamma}) + d^2\sum_{\gamma} \lg^2{\Delta_{\gamma}}
    } \vspace{3pt}\\
    = \sOB\Paren{ 
      d\tau \sum_{\gamma} \lg(q_0^{\gamma}) - d\tau\sum_{\gamma}\lg{\Delta_{\gamma}}
      + d^2 (\sum_{\gamma} \lg(q_0^{\gamma}))^2 
      + d^2 (\sum_{\gamma} \lg{\Delta_{\gamma}})^2
    } \vspace{3pt}\\
    = \sOB\Paren{ 
      d\tau (\tau + \lg{d})
      +d \tau ( d^2 + d \lg{d} + d \tau) 
      + d^2 (\tau^2 + \lg^2{d})
      + d^2 (d^4 + d^2 \tau^2)
    } \vspace{3pt}\\
     = \sOB( d^6 + d^4 \tau^2).
  \end{array}
\end{equation}
}

In the previous equation it possible to write $\sum_{\gamma}
\lg^2{\Delta_{\gamma}} \leq \Paren{\sum_{\gamma}
  \lg{\Delta_{\gamma}}}^2$ because $\Delta_{\gamma} < 1$, and hence
$\lg{\Delta_{\gamma}} < 0$, for all $\gamma$ that are involved in the
sum.  For the roots that holds $\Delta_{\gamma} \geq 1$ the
algorithm isolates them without computing any of their partial
quotients, with the exception of $q_0^{\gamma}$. 
% See also Rem.~\ref{rem:Delta-less-than-1}.

The previous discussion leads to the following theorem.
\begin{theorem}
  \label{th:CF-bound}
  Let $A \in \ZZ[x]$, where $\dg{A}=d$ and $\bitsize{A} = \tau$.
  The worst case complexity of isolating the real roots of $A$ using
  the CF is $\sOB(d^6 + d^4 \tau^2)$.
\end{theorem}

\begin{algorithm2e}[ht]
  %\linesnumbered
  %%  \SetFuncSty{textsc}
  \SetKw{RET}{{\sc return}}
  \SetAlgoVlined
  \KwIn{$A \in \ZZ[X], M(X) = \frac{k X + l}{m X + n}$, $k,l,m,n \in \ZZ$}
  \KwOut{A list of isolating intervals}
  \KwData{Initially $M(X) = X$, i.e. $k=n=1$ and $l=m=0$}

  \BlankLine
  
  \If{ $A(0) = 0$} {
    \FuncSty{OUTPUT}  \FuncSty{Interval}( $M(0), M(0)$) \;
    $A \leftarrow A(X) / X$\;
    \FuncSty{CF}($A, M$)\;
  }
  
  $V \leftarrow \FuncSty{Var}(A)$\;
  \lIf{$V = 0$}{\RET}\;
  \If{$V = 1$}{
    \FuncSty{OUTPUT} \FuncSty{Interval}( $M(0), M(\infty)$)\;
    \RET\;
  }
  
  \BlankLine

  $b \leftarrow \FuncSty{PLB}(A)$ \tcp*[h]{$\FuncSty{PLB} \equiv \FuncSty{PositiveLowerBound}$} \;
  \nllabel{alg:compute:bound}

  \lIf{$b \geq 1$}{
    $A_b \leftarrow A(b+X), M \leftarrow M(b +X)$ \;
    \nllabel{alg:bound:shift}  
  }

  \BlankLine
  
  $A_1 \leftarrow A_b(1+X)$, $M_1 \leftarrow M(1+X)$ \;
  \nllabel{alg:A1}
  \FuncSty{CF}$(A_1, M_1)$ \tcp*[h]{Looking for real roots in $(1, +\infty)$}\;

  $A_2 \leftarrow A_b(\frac{1}{1+X})$, $M_2 \leftarrow M(\frac{1}{1+X})$ \;
  \nllabel{alg:A2}
  \FuncSty{CF}$(A_2, M_2)$ \tcp*[h]{Looking for real roots in $(0, 1)$} \;
  \nllabel{alg:A2transform}
%%   $V_1 \leftarrow \FuncSty{Var}(A_1)$ \;
%%   \If(\tcp*[h]{Possibly real roots in $(0, 1)$}){$V < V_1$}{
%%     $A_2 \leftarrow A(\frac{1}{1+X})$, $M_2 \leftarrow M(\frac{1}{1+X})$ \;
%%     \FuncSty{CF}$(A_2, M_2)$\;
%%   }

  \RET \;
  \caption{\FuncSty{CF}($A, M$)}
  \label{alg:CF}
\end{algorithm2e}

{%\footnotesize
  \bibliographystyle{plain}
  \bibliography{cf}
} %% footnotsize

\end{document}